\documentclass[12pt,a4paper,reqno]{amsart}
\usepackage[all]{xypic}
\usepackage{amsmath,amsfonts,amssymb,amsthm,xcolor,amscd,tikz,xspace,enumerate}
\usepackage[mathscr]{eucal}
\usepackage[pagebackref, colorlinks=true,linkcolor=blue,citecolor=red]{hyperref}
\usepackage{pdfpages}
\usetikzlibrary{arrows}
\usetikzlibrary{arrows.meta}
\usetikzlibrary{intersections}
\usepackage{pgfplots}
\usetikzlibrary{calc,3d,shapes, pgfplots.external, intersections}
\usepackage{pgfplots}
\usepackage{braket}

\usepackage{soul}

\usetikzlibrary{patterns}
\usetikzlibrary{shadings,intersections,calc}

 \newtheorem{theorem}{Theorem}[section]

 \newtheorem{corollary}[theorem]{Corollary}
 
 \newtheorem{proposition}[theorem]{Proposition}
 
 \newtheorem{definition}[theorem]{Definition}
\theoremstyle{remark}
\newtheorem{remark}[theorem]{Remark}
 
 \numberwithin{equation}{section}

\def\.{{\cdot}}

\def\<{\langle} \def\>{\rangle}

\begin{document}

\title[On the Pauli group on 2-qubits in Dynamical Systems with pseudofermions]{On the Pauli group on 2-qubits\\ 
in Dynamical Systems with pseudofermions}
\author[F. Bagarello]{Fabio Bagarello}
\address{Fabio Bagarello\endgraf
Dipartimento di  Ingegneria\endgraf 
Universit\'a  di Palermo\endgraf
Viale delle Scienze, I-90128, Palermo, Italy\endgraf
and \endgraf
Istituto Nazionale di Fisica Nucleare\endgraf
Sezione di Catania, Italy}
\email{fabio.bagarello@unipa.it}

\author[Y. Bavuma]{Yanga Bavuma}
\address{Yanga Bavuma\endgraf
Department of Mathematical Sciences\endgraf 
University of Cape Town\endgraf 
Private Bag X1, 7701, Rondebosch\endgraf 
Cape Town, South Africa}
\email{\texttt{yangabavuma@gmail.com}}

\author[F.G.Russo]{Francesco G. Russo}
\address{Francesco G. Russo\endgraf
Department of Mathematics and Applied Mathematics\endgraf
University of Cape Town \endgraf
Private Bag X1, 7701, Rondebosch\endgraf 
Cape Town, South Africa\endgraf
and \endgraf
Department of Mathematics and Applied Mathematics\endgraf
University of the Western Cape\endgraf
New CAMS Building, Private Bag X17, 7535\endgraf 
Bellville, South Africa} 
\email{\texttt{francescog.russo@yahoo.com}}


\begin{abstract} The group of matrices $P_1$ of Pauli is a finite 2-group of order 16 and plays a fundamental role in quantum information theory, since it is related to the quantum information on the 1-qubit. Here we show that both $P_1$ and the Pauli 2-group $P_2$ of order 64 on 2-qubits, other than in quantum computing, can also appear in dynamical systems which are described by non self-adjoint Hamiltonians. This will allow us to represent $P_1$ and $P_2$ in terms of pseudofermionic operators.
\end{abstract}


\maketitle

\section{Introduction and statement of the results}


Noether's Theorem \cite[Equations (13.148) and (13.158)]{goldstein} is a powerful tool in order to investigate the governing equations in dynamical systems possessing symmetries. The role of groups of symmetries has been largely investigated in connection with  Noether's Theorem and its applications in mathematical physics are numerous. 

In quantum information theory, one is often interested to study quantum codes and, in quantum computing one is interested in preventing or correcting errors that may affect dynamical systems involving qubits. This motivates several authors to work with the Pauli group $P_n$ on $n$ qubits. In particular $P_1$ denotes the Pauli group on 1 qubit
\begin{equation}\label{paulimatrices}
I =\left(\begin{array}{cc} 1 & 0 \\0	 & 1\end{array}\right), \ \ 
X =\left(\begin{array}{cc} 0 & 1 \\1	 & 0\end{array}\right),  \ \
Y=\left(\begin{array}{cc} 0 & -i \\i	 & 0\end{array}\right),  \ \ 
Z=\left(\begin{array}{cc} 1 & 0 \\0	 & -1\end{array}\right),
\end{equation}  
which satisfy the well known  identities
\begin{equation}\label{paulirules}
X=iZY, \ \  Y=iXZ,  \ \ Z=iYX,    \ \  X^2=Y^2=Z^2=I.
\end{equation}

These are the well known Pauli matrices \cite{pauli}, which form the nonabelian  group of order $16$ 
\begin{equation}\label{p1} 
P_1 = \{\pm I, \pm i I, \pm X, \pm iX, \pm Y, \pm iY, \pm Z, \pm iZ\}
\end{equation}
When we deal with quantum information systems, based on more than 1-qubit, one generalizes the notion of Pauli group, considering larger $2$-groups $P_n$ of order $4^{n+1}$ for $n$-qubits with $n \ge 1$ big enough. These groups, known as large Pauli groups, are prominent in the literature of dynamical systems of the last years \cite{kibler, NC}. Here we focus on the case $n=2$, that is, on the Pauli group on 2-qubits, defined by
\begin{equation}\label{p2} 
P_2 = \{M \otimes N \mid M,N \in P_1\},
\end{equation}
where the symbol $\otimes$ denotes the  Kronecker product among matrices. We will realize \eqref{p2} via a new family of examples which come from the area of the electromagnetism, and specifically from the LC quantum circuits.

In this spirit we are going to illustrate the notion of pseudofermionic ladder operators, according to \cite{baginbagbook}; in fact we may reconstruct Hamiltonians from ladder operators as originally indicated by
Dirac and Dieudonn\'e \cite{dieud}. Among other things, we  mention some recent dynamical systems, whose Hamiltonian is  obtained with a construction of the same type, see \cite{bagbavrus, bavuma2021}. 

 We shall mention that \cite{rame1} presents the construction of  a pair of coupled $LC$ electronic circuits, one with amplification and the other with an equivalent attenuation. The interest behind these circuits, besides of the fact that they produce an interesting set of matrices in the analysis of (\ref{p2}), is that they can be used to construct an experimental setting which is connected with a PT symmetric system, i.e. with a quantum system whose Hamiltonian satisfies certain transformation rules with respect to the parity ($P$) and the time reversal ($T$) operators, \cite{ben}. These systems are becoming more and more studied, both by physicists, and by mathematicians, for their very peculiar properties, \cite{ben, specissue2021}. In fact there exists a dynamical system $\mathcal{S}$ in \cite{rame1},  realized by coupled LC electronic circuits, possessing PT symmetries and the dynamics of $\mathcal{S}$ is governed by the differential equation \begin{equation}\label{equationofthemodel}i\frac{\mathrm{d} \Psi(t)}{\mathrm{d} t} = H_\mathcal{S} \Psi (t),\end{equation}where $H_\mathcal{S} = i L_\mathcal{S}$ is the (formal) Hamiltonian of $\mathcal{S}$, the symbols $\alpha, \mu, \gamma$ are parameters defining the circuits, \begin{equation}\label{16}L_\mathcal{S}=\left(\begin{array}{cccc}	0 & 0 & 1 & 0 \\	0 & 0 & 0 & 1 \\	-\alpha & \mu\alpha & \gamma & 0 \\	\mu\alpha &-\alpha & 0 & -\gamma \\\end{array}\right)  \ \qquad \mbox{and} \ \qquad
\Psi(t)=\left(\begin{array}{c}	Q_1(t)  \\	Q_2(t)  \\	\dot Q_1(t)  \\	\dot Q_2(t)  \\\end{array}\right)
\end{equation}is the vector of  charges $Q_1(t)$ and $Q_2(t)$ on the capacitors of the two circuits with corresponding derivatives $\dot Q_1(t)$ and $ \dot Q_2(t)$. It might be useful to stress that (\ref{equationofthemodel}) looks like a Schr\"odinger equation, but this is just a formal identification. In fact, the four components of $\Psi(t)$ are not all independent, as they should.

Our main result is listed below.

\begin{theorem}\label{main1}The dynamical system $\mathcal{S}$ above, realized by coupled LC electronic circuits,  satisfies  the following properties:
\begin{itemize}
\item [(i).] $H_\mathcal{S}$ can be decomposed in the linear combination of 12 bounded operators $X_k$ for $k=1, 2, \ldots, 12$.  Moreover, the operators which commute with all the $X_k$'s are only  those of diagonal type;
\item [(ii).] The   operators  $X_k$  can be identified with matrices in $\mathrm{SL}_4(\mathbb{C})$  and  it is possible to select six operators among the $X_k$ which are generators of the Pauli group $P_2$ on 2-qubits and are simultenously realized by pseudofermionic operators;
\item[(iii).] There exist two subgroups $U$ and $V$ of $P_2$, generated by sets of pseudofermionic operators $\Gamma_\mu$ and $\Gamma_\nu$ respectively, such that $P_2=UV$ and the derived subgroup $[U,V]$ is trivial.
\end{itemize}
\end{theorem}




While Section 2 discusses some general parts of the theory of pseudohermitian operators in dynamical systems and their relevance in Quantum Mechanics, Section 3 focuses specifically on the theory of pseudofermionic operators and on their recent developments. The main proofs can be found in Section 4, while Section 5 describes further investigations which might be  possible for larger Pauli groups.


\section{Some folklore on pseudohermitian physical systems}

Self-adjoint and non-self-adjoint operators are adequate to describe many properties that are useful for physical systems in quantum mechanics, see \cite{goldstein}. Note also that the eigenvalues of  operators $\textbf{a}$, which are significant in many quantum mathematical models, turn out  to be real;   $\langle \psi, \textbf{a} \psi \rangle$ is the expectation value for the measurements of the  observable $\textbf{a}$ in the  quantum state $\psi$, and its eigenvalue $\lambda$ represents one of the possible values for this measurement. 

Wave functions can evolve in time. In fact one of the main  assumptions in quantum mechanics is that there exists an operator $\textbf{H}$ on the Hilbert space $\mathcal{H}$, called the \textit{Hamiltonian operator} for the system, producing  the well known  Schr\"odinger Equation:

\begin{proposition}[See \cite{bchall}, Axiom 5, Claim 3.17]
The \index{time-evolution} time-evolution of the wave function $\psi$ in a quantum system is governed by the  Schr\"{o}dinger Equation,
\begin{equation} \label{schrodeq}
    \frac{d \psi}{dt}=\frac{1}{i \hbar} \textbf{H} \psi,
\end{equation}
where $\textbf{H}$ is the Hamiltonian and $\hbar$ is the  constant of Planck. In particular, if $\textbf{H}$ is  time independent, then \eqref{schrodeq} can be formally solved by setting 
\begin{equation}
    \psi(t)=e^{-i t \textbf{H}/\hbar} \psi_0.
\end{equation}
where $\psi_0$ is the initial condition.
\end{proposition}

Heisenberg had a different view of the dynamics of a quantum system. He thought of the operators \index{observable} (quantum observable) as evolving in time instead of the \index{quantum state} quantum states (vectors in the Hilbert space). In his interpretation each operator $\textbf{a}$ evolves in time according to the operator-valued differential equation
\begin{equation}
    \frac{d \textbf{a}(t)}{dt}=\frac{1}{i \hbar} \left[ \textbf{a}(t), \textbf{H} \right],
\end{equation}
where $\textbf{H}$ is the Hamiltonian operator of the system (time-independent, here, as in the previous proposition), and where  \begin{equation}\label{heisenbergcommutator}\left[ \textbf{a},\textbf{b}\right]=\textbf{a} \textbf{b}-\textbf{b} \textbf{a}\end{equation}
 is \textit{the commutator between the operators } $\textbf{a}$ and $\textbf{b}$   according to \cite[ Definition 3.20]{bchall}. Note that since $\textbf{H}$ commutes with itself, the operator $\textbf{H}$ remains constant in time, even in this interpretation. However this turns out to have the same physical meaning in  Schr\"{o}dinger's interpretation.

An increasing number of physicists and mathematicians started to be recently interested to situations in which the Hamiltonian of a mathematical model under consideration is not necessarily  self-adjoint. Examples of a non-self-adjoint Hamiltonian operators  have been studied by Bender \cite{ben,bender3,bender2,bender1}), Mostafazadeh  \cite{mostaf2,mostaf1} and other authors \cite{bagrus1, bagrus2, bagrus3}, introducing, for instance, generalized versions of the  harmonic oscillator. 

\section{Elementary theory of pseudofermionic and pseudobosonic operators}
In particle physics, elementary particles and composite particles are essentially broken up into two (non-overlapping) classes:  \textit{bosons} and \textit{ fermions}. Roughly speaking  fermions are particles that are associated with the matter and  \textit{pseudofermions} are generalizations of  fermions. The bosons are usually used to represent the interactions (in terms of fields),  while \textit{pseudobosons} generalize bosons. In this paper, we will only focus on pseudofermions. The reader can find more information the theory of the pseudobosons and of the pseudofermions  in \cite{baginbagbook, bag2022book}.

Typically on a two dimensional  Hilbert space $\mathcal{H}=\mathbb{C}^2$  we can define  \textit{lowering} and  \textit{raising operators}, which will lower or raise the eigenvalues (respectively) associated with an eigenstate by acting on the state itself. These come as a pair, say $\textbf{c}$ and $\textbf{c}^*$ respectively. The lowering operator, $\textbf{c}$ lowers the eigenvalue of a given state by acting on it, and the raising operator $\textbf{c}^*$ raises the eigenvalue of a given state by acting on it, and it is the adjoint of the lowering operator. 

\begin{remark} The  \textit{fermionic operators} satisfy the \textit{ canonical anticommutation relations} (CAR) \begin{equation}\label{ccr}\{\textbf{c},\textbf{c}^*\}:=\textbf{c}\textbf{c}^*+\textbf{c}^*\textbf{c}=\mathbb{I}, \ \ \{\textbf{c},\textbf{c}\}=\{\textbf{c}^*,\textbf{c}^*\}=0.\end{equation} However  \textit{pseudofermionic operators} satisfy more general anticommutation  rules:
\begin{equation} \label{pseudofermions}
    \{\textbf{a},\textbf{b}\}=\mathbb{I}, \ \ \{\textbf{a},\textbf{a}\}=\{\textbf{b},\textbf{b}\}=0,
\end{equation}
where $\textbf{b} \neq \textbf{a}^*$ a priori. Note that \eqref{ccr} is motivated by Pauli's Principle.
\end{remark}

Note that  fermionic operators are bounded. A nonzero vector $\varphi_0 \in \mathcal{H}$ such that $\textbf{a}\varphi_0=0$ surely exists, as well as a nonzero vector $\Psi_0 \in \mathcal{H}$ such that $\textbf{b}^*\Psi_0=0$. This is because $\textbf{a}$ and $\textbf{b}^*$ have nontrivial  kernels. We are going to illustrate better this aspect, providing a {\em minimal} framework of functional analysis in the remaining part of the present section.

\begin{definition}A set $\mathcal{E}=\{e_n \in \mathcal{H}, n \geq 0\}$ is a Schauder basis for $\mathcal{H}$ if any vector $f \in \mathcal{H}$ can be written uniquely as  \[f= \sum\limits^\infty_{n=0} c_n (f) e_n,\] that is, as   linear combination of $e_n$ (eventually  infinite) with $c_n (f) \in \mathbb{C}$ depending only on $f$.\end{definition}

A particular  type of bases is given by the \textit{orthonormal bases}, that is, Schauder bases such that $\langle e_n,e_m \rangle = \delta_{n,m}$, where $\delta_{n,m}$ denotes the Kronecker delta and $n,m$ are positive integers.  In this particular case, we recover $c_n (f)= \langle e_n, f \rangle$ via the scalar product. Note  that  the scalar product is  linear in its second variable. A more specific version of  orthonormal bases is  given by the  \textit{Riesz bases}:

\begin{definition}\label{rieszbases} In $\mathcal{H}$ the set $\mathcal{F}=\{f_n \in \mathcal{H}, n \geq 0\}$ is a Riesz basis if there exists a  bounded operator $T$ on $\mathcal{H}$  with bounded inverse $T^{-1}$ and an orthonormal basis $\mathcal{E}=\{e_n \in \mathcal{H}, n \geq 0\}$ such that $f_n=T e_n$, for all $n \geq 0$. \end{definition}

A detailed analysis on operators defined by Riesz bases can be found in \cite{bagino,christ1}, but it is clear that $\langle f_n, f_m \rangle \neq \delta_{n,m}$ in the context of Definition \ref{rieszbases}. In this case, however, one can introduce a second set $\mathcal{G}$, which is another Riesz basis biorthonormal to $\mathcal{F}$, whose vectors are simply $g_n=(T^{-1})^*e_n$. We go ahead to illustrate some results that are crucial in the present framework. 

First we define the following vectors
\begin{equation}\label{fa1}
    \varphi_1 := \textbf{b}\varphi_0, \ \ \Psi_1 := \textbf{a}^* \Psi_0,
\end{equation}
as well as the following \index{operator!non-self-adjoint} non-self-adjoint operators
\begin{equation} \label{fa2}
    \textbf{N} := \textbf{b}\textbf{a}, \ \ \textbf{N}^*=\textbf{a}^* \textbf{b}^*.
\end{equation}
Note that for $n \geq 2$ the vectors $\textbf{b}^n \varphi_0$ and ${\textbf{a}^*}^n \Psi_0$ are automatically equal to zero. 
Then the following equations are satisfied:
\begin{equation}\label{fa3}
    \textbf{a}\varphi_1=\varphi_0, \ \ \textbf{b}^*\Psi_1=\Psi_0
\end{equation}
\begin{equation}\label{fa4}
    \textbf{N}\varphi_n=n\varphi_n, \ \ \textbf{N}^*\Psi_n=n\Psi_n, \text{ for $n=0,1$}.
\end{equation}
If the normalizations of $\varphi_0$ and $\Psi_0$ are chosen such that $\langle \varphi_0, \Psi_0 \rangle=1$, then we have also
\begin{equation}\label{fa5}
    \langle \varphi_k, \Psi_n \rangle=\delta_{k,n} \text{ for $k,n=0,1$}.
\end{equation}
Now we introduce the \index{operator!self-adjoint} self-adjoint operators $\textbf{S}_\varphi$ and $\textbf{S}_\Psi$ via their action on a generic $f \in \mathcal{H}$:
\begin{equation}\label{fa6}
	\textbf{S}_\varphi f= \sum\limits^1_{n=0} \langle \varphi_n, f \rangle \varphi_n, \ \ \textbf{S}_\Psi f= \sum\limits^1_{n=0} \langle \Psi_n, f \rangle \Psi_n.
\end{equation}
Note also that the operators $\textbf{S}_\varphi$ and $\textbf{S}_\Psi$ are bounded, strictly positive, self-adjoint, and invertible. They satisfy
\begin{equation}\label{fa7}
    \|\textbf{S}_\varphi\| \leq \| \varphi_0\|^2+\|\varphi_1\|^2, \ \ \|\textbf{S}_\Psi\| \leq \| \Psi_0\|^2+\|\Psi_1\|^2,
\end{equation}
\begin{equation}\label{fa8}
    \textbf{S}_\varphi \Psi_n=\varphi_n, \ \ \textbf{S}_\Psi\varphi_n=\Psi_n,
\end{equation}
for $n=0,1$, as well as $\textbf{S}_\varphi=\textbf{S}_\Psi^{-1}$ and the following intertwining relations
\begin{equation}\label{fa9}
    \textbf{S}_\Psi \textbf{N}=\textbf{N}^*\textbf{S}_\Psi, \ \ \textbf{S}_\varphi \textbf{N}^*=\textbf{N}\textbf{S}_\varphi.
\end{equation}
The vectors of  $\mathcal{F}_\varphi=\{\varphi_0, \varphi_1\}$ and $\mathcal{F}_\Psi=\{\Psi_0, \Psi_1\}$ are  biorthogonal,  linearly independent in a two dimensional complex Hilbert space so that $\mathcal{F}_\varphi$ is a Schauder basis for $\mathcal{H}$. The same argument applies to $\mathcal{F}_\Psi$ and in fact it turns out that  both these sets are  Riesz bases for $\mathcal{H}$.


A  connection between pseudofermions and  fermions is reported below:

\begin{proposition}[See \cite{baginbagbook}, Theorem 3.5.1]\label{similarity} Let $\textbf{c}$ and $\textbf{T}=\textbf{T}^{*}$ be two operators on $\mathcal{H}$ such that \eqref{ccr} are satisfied and in addition $\textbf{T}$ is positive. Then the operators 
$\textbf{a}=\textbf{TcT}^{-1}$ and $\textbf{b}=\textbf{Tc}^{*}\textbf{T}^{-1}$
satisfy \eqref{pseudofermions}.
Viceversa given two operators $\textbf{a}$ and $\textbf{b}$ acting on $\mathcal{H}$, satisfying \eqref{pseudofermions}, it is possible to construct two operators $\textbf{c}$ and $\textbf{T}$ with the above properties.
\end{proposition}

Examples of  pseudofermions (and generalizations) can be found in  \cite{bgsa,bagrus1,bagrus2,bagrus3,baginbagbook,bagpan}.

 In  particular \cite[Chapter 3.5.1]{baginbagbook} shows an effective  non-selfadjoint hamiltonian of a mathematical model of a two-level atom which  interacts with an electromagnetic field. Let's recall a few facts from \cite[Chapter 3.5.1]{baginbagbook} with more details; it will be useful in the proof of Theorem \ref{main1}. Maamache and others \cite{tripf} described an effective non self-adjoint Hamiltonian involved in a two levels atom interacting with an electromagnetic field. This  followed an intuition of Mostafazadeh in \cite{mostaf2, mostaf1}. Writing \eqref{schrodeq}, one gets \begin{equation}\label{specialschrodeq} i\dot\Phi(t)=H_{eff}\Phi(t), \qquad H_{eff}=\frac{1}{2}\left(
                                                       \begin{array}{cc}
                                                         -i\delta & \overline{\omega} \\
                                                         \omega & i\delta \\
                                                       \end{array}
                                                     \right),
\end{equation} 
where $\delta$ is a real quantity, related to the decay rates for the two levels, while the complex parameter $\omega$ describes the interaction due to radiations of the atom. In particular we may write \begin{equation}\label{omega}\omega=|\omega|e^{i\theta}
\end{equation}
 Of course, the effective Hamiltonian cannot be self-adjoint in the present situation, but let's see better some details in connection with the framework of functional analysis which we have seen in \eqref{fa1}-\eqref{fa9}. Now introduce the operators
\begin{equation}\label{firstoperator}
	\textbf{a}=\frac{1}{2\Omega}\left(
	\begin{array}{cc}
		-|\omega| & -e^{-i\theta}(\Omega+i\delta) \\
		e^{i\theta}(\Omega-i\delta) & |\omega| \\
	\end{array}
	\right), \quad
	\textbf{b}=\frac{1}{2\Omega}\left(
	\begin{array}{cc}
		-|\omega| & e^{-i\theta}(\Omega-i\delta) \\
		-e^{i\theta}(\Omega+i\delta) & |\omega| \\
	\end{array}
	\right),
\end{equation}
where $\Omega=\sqrt{|\omega|^2-\delta^2}$ can be assumed to be real and strictly positive. One can check easily that \eqref{firstoperator} are pseudofermionic operators, that is, they satisfy \eqref{pseudofermions} and that
\begin{equation}\label{effectivehamiltonian}H_{eff}=\Omega\left(\textbf{b}\textbf{a}-\frac{1}{2}\mathbb{I}\right).
\end{equation}

In order to visualize the notions in \eqref{fa1}--\eqref{fa9}, we may consider 
\begin{equation}\label{fa10}
\varphi_0=k\left(
             \begin{array}{c}
               1 \\
               -\,\frac{e^{i\theta}(\Omega-i\delta)}{|\omega|} \\
             \end{array}
           \right),\qquad
\Psi_0=k'\left(
             \begin{array}{c}
               1 \\
               -\,\frac{e^{i\theta}(\Omega+i\delta)}{|\omega|} \\
             \end{array}
           \right),
\end{equation}
where $k$ and $k'$ are normalization constants such that  the orthogonality condition below is satisfied
\begin{equation}\label{fa11}\left<\varphi_0,\Psi_0\right>=\overline{k}\,k'\left(1+\frac{1}{|\omega|^2}(\Omega+i\delta)^2\right)=1.
\end{equation} Then we introduce
\begin{equation}\label{fa12}
\varphi_1=\textbf{b}\varphi_0=k\left(
             \begin{array}{c}
               \frac{i\delta-\Omega}{|\omega|} \\
               -e^{i\theta} \\
             \end{array}
           \right),\qquad
\Psi_1=\textbf{a}^\dagger\Psi_0=k'\left(
             \begin{array}{c}
               \frac{-i\delta-\Omega}{|\omega|} \\
               -e^{i\theta} \\
             \end{array}
           \right).
\end{equation}
In particular, we find that $\mathcal{F}_\varphi$ and $\mathcal{F}_\Psi$ are biorthonormal bases of $\mathcal{H}$, and that
\begin{equation}\label{fa13}
H_{eff}\varphi_0=-\,\frac{\Omega}{2}\,\varphi_0,\quad H_{eff}\varphi_1=\frac{\Omega}{2}\,\varphi_1, \quad
 H_{eff}^\dagger\Psi_0=-\,\frac{\Omega}{2}\,\Psi_0,\quad H_{eff}^\dagger\Psi_1=\frac{\Omega}{2}\,\Psi_1.
\end{equation}
It is evident that $H_{eff}$ and $H_{eff}^\dagger$ are not self-adjoint. In order to find \eqref{fa8} and \eqref{fa9}, now
\begin{equation}\label{fa14}
\mathbf{S}_\varphi=2|k|^2\left(
                  \begin{array}{cc}
                    1 & \frac{-i\delta}{|\omega|}\,e^{-i\theta} \\
                    \frac{i\delta}{|\omega|}\,e^{i\theta} & 1 \\
                  \end{array}
                \right),\quad
\mathbf{S}_\Psi=\frac{|\omega|^2}{2|k|^2\Omega^2}\left(
                  \begin{array}{cc}
                    1 & \frac{i\delta}{|\omega|}\,e^{-i\theta} \\
                    \frac{-i\delta}{|\omega|}\,e^{i\theta} & 1 \\
                  \end{array}
                \right)
\end{equation}
and  one is the inverse of the other. In the present situation Proposition \ref{similarity} may be applied to the operator $\mathbf{S}_\varphi^{\pm1/2}=\pm\frac{1}{2} \mathbf{S}_\varphi$ to define two {\em standard} fermionics operators  $\textbf{c}$ and $\textbf{c}^\dagger$ such that 
\begin{equation}\label{fa15}
\mathbf{N}_0=\textbf{c}^\dagger \textbf{c} \  \ \mbox{and}  \ \ H_{eff}=\mathbf{S}_\varphi^{1/2}   \ \left(\Omega \left(\mathbf{N}_0 - \frac{1}{2} \mathbb{I} \right)\right) \ \ \mathbf{S}_\varphi^{-1/2},
\end{equation}
but now we have clearly produced another Hamiltonian which is self-adjoint and similar to $H_{eff}$.

More recently, \cite[Theorem 1.2]{bagbavrus} describes the Pauli group \eqref{paulimatrices} in terms of \eqref{firstoperator}. In fact, it is possible to introduce the matrices 
\begin{equation}\label{muoperators}
	\mu_1=\frac{1}{\Omega}\left(
	\begin{array}{cc}
		-|\omega| & -i\delta e^{-i\theta} \\
		-i\delta e^{i\theta} & |\omega| \\
	\end{array}
	\right), \quad \mu_2=i\left(
	\begin{array}{cc}
		0 &  e^{-i\theta} \\
		e^{i\theta} & 0 \\
	\end{array}
	\right), \quad \mu_3=\frac{1}{\Omega}\left(
	\begin{array}{cc}
		i\delta & -|\omega| e^{-i\theta} \\
		-|\omega| e^{i\theta} & -i\delta \\
	\end{array}
	\right),
\end{equation}
where
\begin{equation} \label{mumatrices}
	\mu_1=\mu_1(\theta, \delta)=\textbf{b}+\textbf{a}, \qquad \mu_2=\mu_2(\theta, \delta)=i(\textbf{b}-\textbf{a}), \qquad \mu_3=\mu_3(\theta, \delta)=[\textbf{a},\textbf{b}]=\textbf{a}\textbf{b}-\textbf{b}\textbf{a}.
\end{equation} 
Of course, $\mu_1=\mu_1(\theta, \delta)$, that is, $\mu_1$ depends on $\theta$ and $\delta$, similarly this happens for $\mu_2$ and $\mu_3$, but also   for $\textbf{a}=\textbf{a}(\theta, \delta)$ and $\textbf{b}=\textbf{b}(\theta, \delta)$ and their linear combinations. As shown in \cite[Theorem 1.2]{bagbavrus} independently on the choice of $\theta$ and $\delta$ we may consider the set \begin{equation}\label{pmu}P_\mu=\{\mu_1,\mu_2,\mu_3\}
\end{equation} which is a concrete realization of the generators of the first Pauli group $P_1$ via pseudofermionic operators.  On the other hand,  specific choices of $\theta$ and $\delta$ give more precise proportionality relations between \eqref{muoperators} and \eqref{paulirules}. For instance, the choice   $\theta=\pi/2$ and $\delta=0$  allows us to consider
\begin{equation}\label{specialmuoperators1}
	\mu_1 \left( \frac{\pi}{2}, 0 \right)=-Z, \quad \mu_2 \left(\frac{\pi}{2},  0 \right)=-iY, \quad \mu_3 \left(\frac{\pi}{2},  0 \right)=-Y,
\end{equation}
and using $X=iZY$ in \eqref{paulirules} we find that 
\begin{equation}\label{specialmuoperators2}
	X=  \mu_1 \left(\frac{\pi}{2},0 \right) \  \mu_2 \left(\frac{\pi}{2},0 \right), \quad Y=-\mu_3 \left( \frac{\pi}{2},  0 \right), \quad Z=-\mu_1 \left(\frac{\pi}{2},  0 \right).
\end{equation}
Of course, we may use \eqref{firstoperator} and \eqref{mumatrices}, getting in correspondence with   $\theta=\pi/2$ and $\delta=0$
\begin{equation}\label{specialmuoperators3}
	X=  i(\textbf{a}\textbf{b} - \textbf{b}\textbf{a}) \left(\frac{\pi}{2},   0 \right), \quad Y= (\textbf{b}\textbf{a}-\textbf{a}\textbf{b}) \left( \frac{\pi}{2},  0 \right), \quad Z=-(\textbf{b} + \textbf{a}) \left(\frac{\pi}{2}, 0 \right).
\end{equation}
The above computations play a fundamental role in the proof of Theorem \ref{main1}.

\section{Proof of the  main theorem}


Before we prove our main result, we recall separately some terminology from \cite{robinson}:

\begin{definition}[See \cite{robinson}, p.145]\label{centralproduct} A group $G$ is (the internal) central product of its subgroups $H$ and $K$, if $G=HK$ and the derived subgroup of $H$ and $K$ is trivial, that is, \begin{equation}\label{groupcommutator}
[H,K]=\{h^{-1}k^{-1}hk \mid h \in H, k \in K \}=1.
\end{equation} 
\end{definition}

For groups as in Definition \ref{centralproduct}, one has that both the subgroups realizing the central product are normal in the group, that is,  $H$ and $K$ are normal in $G$. Incidentally, \eqref{groupcommutator} is notationally similar to \eqref{heisenbergcommutator}. While in the first case one uses the algebraic structure of group (in multiplicative notation), in order to produce a new object which still has the structure of group, in the second case one uses two functional operators, in order to produce a new functional operator.


\begin{proof}[Proof of Theorem \ref{main1}]

(i). Let us consider the  set of matrices $\{X_j, j=1,2,\ldots,12\}$, where
\begin{equation}\label{matricesx}
X_1=\left(
\begin{array}{cccc}
	0 & 0 & 1 & 0 \\
	0 & 0 & 0 & 1 \\
	1 & 0 & 0 & 0 \\
	0 & 1 & 0 & 0 \\
\end{array}
\right),  X_2=\left(
\begin{array}{cccc}
	0 & 0 & -i & 0 \\
	0 & 0 & 0 & -i \\
	i & 0 & 0 & 0 \\
	0 & i & 0 & 0 \\
\end{array}
\right),
 X_3=\left(
\begin{array}{cccc}
	1 & 0 & 0 & 0 \\
	0 & 1 & 0 & 0 \\
	0 & 0 & -1 & 0 \\
	0 & 0 & 0 & -1 \\
\end{array}
\right),\end{equation}
\begin{equation}  X_4=\left(
\begin{array}{cccc}
	1 & 0 & 0 & 0 \\
	0 & -1 & 0 & 0 \\
	0 & 0 & 1 & 0 \\
	0 & 0 & 0 & -1 \\
\end{array}
\right), X_5=\left(
\begin{array}{cccc}
	0 & 1 & 0 & 0 \\
	1 & 0 & 0 & 0 \\
	0 & 0 & 0 & 1 \\
	0 & 0 & 1 & 0 \\
\end{array}
\right),  X_6=
\left(
\begin{array}{cccc}
		0 & 1 & 0 & 0 \\
	-1 & 0 & 0 & 0 \\
	0 & 0 & 0 & 1 \\
	0 & 0 & -1 & 0 \\
\end{array}
\right),
\end{equation}
\begin{equation}
X_7=\left(
\begin{array}{cccc}
	0 & 0 & 0 & 1 \\
	0 & 0 & 1 & 0 \\
	0 & 1 & 0 & 0 \\
	1 & 0 & 0 & 0 \\
\end{array}
\right),  X_8=\left(
\begin{array}{cccc}
	0 & 0 & -i & 0 \\
	0 & 0 & 0 & i \\
	i & 0 & 0 & 0 \\
	0 & -i & 0 & 0 \\
\end{array}
\right), X_9=\left(
\begin{array}{cccc}
	0 & 0 & 0 & -i \\
	0 & 0 & -i & 0 \\
	0 & i & 0 & 0 \\
	i & 0 & 0 & 0 \\
\end{array}
\right), 
\end{equation}
\begin{equation}
 X_{10}=\left(
\begin{array}{cccc}
	1 & 0 & 0 & 0 \\
0 & -1 & 0 & 0 \\
0 & 0 & -1 & 0 \\
0 & 0 & 0 & 1 \\
\end{array}\right), X_{11}=\left(
\begin{array}{cccc}
	0 & 1 & 0 & 0 \\
	1 & 0 & 0 & 0 \\
	0 & 0 & 0 & -1 \\
	0 & 0 & -1 & 0 \\
\end{array}\right),  
X_{12}=\left(
\begin{array}{cccc}
	0 & 0 & 1 & 0 \\
0 & 0 & 0 & -1 \\
1 & 0 & 0 & 0 \\
0 & -1 & 0 & 0 \\
\end{array}
\right).\end{equation}

Now observe that \begin{equation}\label{matricesxbis}[X_j,X_k]=X_jX_k-X_kX_j \ \  \ \ \mbox{and} \ \ \ \ \{X_j,X_k\}=X_jX_k+X_kX_j.
\end{equation} 
The matrices $X_j$ satisfy a mixed set of rules. For instance we have \begin{equation}\{X_1,X_2\}=\{X_1,X_3\}=0  \ \ \mbox{and} \ \ [X_2,X_4]=[X_2,X_5]=0.\end{equation} This aspect of the matrices is relevant when deducing explicitly the equation of motion, \cite{mer,rom}. To calculate the matrices that commute with all the $X_i$'s we consider a general expression for the following matrix below with  coefficients $f, g, h, \ldots, u \in \mathbb{C}$
\begin{equation}M=\begin{pmatrix}
f &g &h &d
\\
j &k &l &m
\\
n &o &p &q
\\
r &s &t &u
\end{pmatrix}\end{equation}
and calculate the commutator of the matrix $M$ with each of the $X_i$'s.
From a detailed (and long) analysis of the various $[M,X_j]$ it follows that the only matrices which commute with all the $X_j$'s are those proportional to the identity matrix. In fact we have to solve the linear system of 192 (12 $\times$ 16) equations in 16 complex variables. 
Of course, all the operators $X_i$ turns out to be bounded, since they are expressed in terms of finite dimensional complex matrices.

It is now an easy computation to check that $L_\mathcal{S}$ in (\ref{16}) can be written as a linear combination of the $X_j$'s above. In particular, we can check that
\begin{equation}
L_\mathcal{S}=\sum_{k=1}^{12}\alpha_k X_k,
\end{equation}
where the only nonzero $\alpha_k$ are the following:
\begin{equation}
\alpha_1=\frac{1-\alpha}{2}, \quad \alpha_2=\frac{i(1+\alpha)}{2},\quad \alpha_4=\frac{\gamma}{2}, \quad \alpha_7=\frac{\alpha\mu}{2}, \quad \alpha_9=-\frac{\alpha\mu}{2i}, \quad \alpha_{10}=-\frac{\gamma}{2}.
\end{equation}
Note that the solution of the differential equation for $\Psi(t)$ in \eqref{equationofthemodel} turns out to be
\begin{equation}
\Psi(t)=e^{L_\mathcal{S} \ t}\Psi(0),
\end{equation}
which of course can be rewritten in terms of the $X_j$'s. Therefore, recalling that $H_\mathcal{S}=iL_\mathcal{S}$, (i) is proved completely.

(ii). First of all, one can check from \eqref{matricesx} that all the matrices $X_k$ belongs to $\mathrm{SL}_4(\mathbb{C})$.

Now the idea is that we begin to construct two pseudofermionic operators $\textbf{a}$ and $\textbf{b}$ on  $\mathbb{C}^2$ satisfying  \eqref{pseudofermions} and realizing \eqref{p1}.  
The various matrices $X_j$ can be written in terms of pseudofermionic operators $\textbf{A}$, $\textbf{B}$, $\widetilde{\textbf{A}}$ and $\widetilde{\textbf{B}}$, which are deduced as tensor product of matrices in $\mathrm{SL}_2(\mathbb{C})$. Interestingly enough, $\textbf{A}$, $\textbf{B}$, $\widetilde{\textbf{A}}$ and $\widetilde{\textbf{B}}$ can be deduced from some matrices, which appear  in the analysis of the two levels atom interacting with an electromagnetic field as in  Section 3.  Here we may use ${\mathbf a}$ and ${\mathbf b}$ as in \eqref{firstoperator} and introduce
\begin{equation} \label{53}
\textbf{A}:=\textbf{a} \otimes \mathbb{I}, \qquad \textbf{B}:=\textbf{b} \otimes \mathbb{I}.
\end{equation}
It is easy to check that also $\textbf{A}$ and $\textbf{B}$ give a concrete realization for $P_\mu$, but now $\textbf{A}$ and $\textbf{B}$
are expressed by matrices in $\mathrm{SL}_4(\mathbb{C})$ and no longer by matrices in $\mathrm{SL}_2(\mathbb{C})$  as it was for $\textbf{a}$ and $\textbf{b}$.  We have created in this way two pseudofermionic operators $\textbf{A}$ and $\textbf{B}$ on the Hilbert space $\mathbb{C}^4$, via the notion of Kronecker product from \eqref{firstoperator}. We can find another two pseudofermionic operators 
\begin{equation} \label{aandbtilde}
\widetilde{\textbf{A}}:= \mathbb{I} \otimes \textbf{a}, \qquad \widetilde{\textbf{B}}:= \mathbb{I} \otimes \textbf{b}.
\end{equation}
always in the same  Hilbert space $\mathbb{C}^4$, where we described $\textbf{A}$ and $\textbf{B}$. In fact, one can check that also \eqref{aandbtilde} satisfy \eqref{pseudofermions}. We may conclude  that the set \begin{equation}\label{generatingsetone}
\Gamma_\mu = \{\mu_1 \otimes \mathbb{I},  \mu_2 \otimes \mathbb{I}, \mu_3 \otimes \mathbb{I}\}
\end{equation}
represents a concrete realization for $P_1$ via the operators $\textbf{A}$ and $\textbf{B}$ in $\mathbb{C}^4$, moreover in the same Hilbert space we find also the set 
\begin{equation}\label{generatingsetwo}
\Gamma_\nu=\{\mathbb{I} \otimes \mu_1, \mathbb{I} \otimes \mu_2,  \mathbb{I}\otimes \mu_3\}
\end{equation}
which represents another concrete realization for $P_1$, but via the operators $\widetilde{\textbf{A}}$ and $\widetilde{\textbf{B}}$.
Note  that  the Pauli matrices  $X$, $Y$ and $Z$ in \eqref{paulimatrices} may be expressed by \eqref{specialmuoperators2} and \eqref{specialmuoperators3}. Consequently,  $\textbf{A}=\textbf{A}(\theta, \delta)$, $\textbf{B}=\textbf{B}(\theta, \delta)$, $\widetilde{\textbf{A}} = \widetilde{\textbf{A}} (\theta, \delta)$ and $\widetilde{\textbf{B}} = \widetilde{\textbf{B}}(\theta, \delta)$, where $\theta$ and $\delta$  are introduced in \eqref{omega}. Therefore we get
\begin{equation}\label{mutensorproduct}
\mu_1 \otimes \mathbb{I}=\textbf{B}+\textbf{A}, \ \mu_2 \otimes \mathbb{I} =i (\textbf{B}-\textbf{A}),  \ \mu_3 \otimes \mathbb{I} =\textbf{A}\textbf{B}-\textbf{B}\textbf{A},
\end{equation}
\[ \mathbb{I} \otimes \mu_1=  \widetilde{\textbf{B}} + \widetilde{\textbf{A}}, \ \mathbb{I} \otimes \mu_2=  i (\widetilde{\textbf{B}} - \widetilde{\textbf{A}}), \    \mathbb{I} \otimes \mu_3 =\widetilde{\textbf{A}} \widetilde{\textbf{B}}-\widetilde{\textbf{B}}\widetilde{\textbf{A}}.\]
We are going to report some computations for the matrices $X_1, X_2, X_3, X_4, X_5, X_6$. In fact we may write in correspondence with $\theta=\pi/2$ and $\delta=0$
\begin{equation}\label{x1asmus}
X_1=X \otimes \mathbb{I}=\left(\mu_1\left(\frac{\pi}{2},0\right)\mu_2\left(\frac{\pi}{2},0\right)\right) \otimes \mathbb{I}=(i (\textbf{a}\textbf{b}-\textbf{b}\textbf{a}) \otimes \mathbb{I}) \left(\frac{\pi}{2},0\right)
\end{equation}
\[=i ((\textbf{a} \otimes \mathbb{I}) (\textbf{b} \otimes \mathbb{I}) -(\textbf{b} \otimes \mathbb{I}) (\textbf{a}\otimes \mathbb{I})) \left(\frac{\pi}{2},0\right)=i(\textbf{A}\textbf{B}-\textbf{B}\textbf{A})\left(\frac{\pi}{2},0\right),\]
\begin{equation}  X_2=Y \otimes \mathbb{I} = -\mu_3\left(\frac{\pi}{2},0\right) \otimes \mathbb{I}= (\textbf{B}\textbf{A}-\textbf{A}\textbf{B}) \left(\frac{\pi}{2},0\right),  
\end{equation}
\begin{equation} X_3=Z \otimes \mathbb{I}=-\mu_1\left(\frac{\pi}{2},0\right) \otimes \mathbb{I}=-(\textbf{B}+\textbf{A})\left(\frac{\pi}{2},0\right),
\end{equation}    \begin{equation} X_4=  \mathbb{I} \otimes Z = \mathbb{I} \otimes \left(-\mu_1\left(\frac{\pi}{2},0\right)\right)  = - (\widetilde{\textbf{B}} + \widetilde{\textbf{A}}) \left(\frac{\pi}{2},0\right), 
\end{equation}
\begin{equation}\label{x5asmus} X_5=\mathbb{I} \otimes X =   \mathbb{I} \otimes \left(\mu_1\left(\frac{\pi}{2},0\right)\mu_2 \left(\frac{\pi}{2},0\right)\right)   = \mathbb{I} \otimes (i (\textbf{a}\textbf{b}-\textbf{b}\textbf{a})) \left(\frac{\pi}{2},0\right)
\end{equation} 
\[=i (( \mathbb{I}\otimes \textbf{a}) ( \mathbb{I} \otimes \textbf{b} ) -( \mathbb{I} \otimes \textbf{b} ) ( \mathbb{I} \otimes \textbf{a} )) \left(\frac{\pi}{2},0\right)=i(\widetilde{\textbf{A}}\widetilde{\textbf{B}}  - \widetilde{\textbf{B}}\widetilde{\textbf{A}}) \left(\frac{\pi}{2},0\right),\]
\begin{equation}\label{x6}  X_6=i (\mathbb{I} \otimes Y) =   i \left(\mathbb{I} \otimes \left(-\mu_3\left(\frac{\pi}{2},0\right)\right)\right)  = i (\widetilde{\textbf{B}}\widetilde{\textbf{A}} - \widetilde{\textbf{A}}\widetilde{\textbf{B}}) \left(\frac{\pi}{2},0\right).  
\end{equation} 

This means that $\{X_1, X_2, X_3, X_4, X_5, X_6\}$ is a set, which is realized by pseudofermionic operators. Now we claim that the same set is also a set of generators for $P_2$. 

Consider that \begin{equation}\label{uandv}U  \ \mbox{is the Pauli group generated by} \  \Gamma_\mu  \ \ \ \mbox{and} \ \ V  \ \mbox{is the Pauli group generated by}   \ \Gamma_\nu, 
\end{equation} referring to  \eqref{generatingsetone} and  \eqref{generatingsetwo}. Of course, at the level of isomorphism of groups we have
\begin{equation}U \simeq V \simeq P_1
\end{equation} and both $U$ and $V$ are subgroups of $\mathrm{SL}_4(\mathbb{C})$. Having in mind the usual matrix product in $\mathrm{SL}_4(\mathbb{C})$ and the rules in \eqref{paulirules} in terms of \eqref{mutensorproduct}, we have that for \eqref{mutensorproduct} with $\theta=\pi/2$ and $\delta=0$
\begin{equation} \label{decompositionrule}
P_2 = UV=\langle \mu_1 \otimes \mathbb{I},  \mu_2 \otimes \mathbb{I}, \mu_3 \otimes \mathbb{I} \rangle  \ \langle  \mathbb{I} \otimes \mu_1, \mathbb{I} \otimes \mu_2,  \mathbb{I}\otimes \mu_3  \rangle\end{equation} 
\begin{equation} =\langle \mu_1 \otimes \mathbb{I},  \mu_2 \otimes \mathbb{I}, \mu_3 \otimes \mathbb{I}, \mathbb{I} \otimes \mu_1, \mathbb{I} \otimes \mu_2,  \mathbb{I}\otimes \mu_3 \rangle = \langle X_1, X_2, X_3, X_4, X_5, X_6 \rangle.
\end{equation}
Note that \eqref{decompositionrule} reflects at the level of operators the fact that each element of $P_2$ may be written as $M \otimes N=(M \otimes \mathbb{I})(\mathbb{I} \otimes N)$ with $M$ and $N$ which are Pauli matrices, see \eqref{p2}. Consequently, each element of $P_2$ can be written as product of  elements of $\Gamma_\mu$ and $\Gamma_\nu$, that is, $\Gamma_\mu \cup \Gamma_\nu$ generates $P_2$. This means that the claim is true in particular for  \eqref{specialmuoperators1}, \eqref{specialmuoperators2}, \eqref{specialmuoperators3}, hence $\{X_1, X_2, X_3, X_4, X_5,X_6\}$ is a set of generators for $P_2$ in terms of pseudofermionic operators.

(iii). From what we have seen in (ii) above, $P_2=UV$, where $U$ and $V$ are in \eqref{uandv}. We have exactly the  conditions in Definition \ref{centralproduct} if we look at \eqref{uandv}. In fact, the derived subgroup of $U$ and $V$ (according to \eqref{groupcommutator}) is given by \begin{equation}\label{uandvbis}[U,V]=[\langle \mu_1 \otimes \mathbb{I},  \mu_2 \otimes \mathbb{I}, \mu_3 \otimes \mathbb{I} \rangle  , \langle  \mathbb{I} \otimes \mu_1, \mathbb{I} \otimes \mu_2,  \mathbb{I}\otimes \mu_3 \rangle]
\end{equation} 
and it is enough to check that the generators of $U$ and $V$ commute at the group level, in order to conclude that $[U,V]$ is trivial. Now if $i,j \in \{1, 2, 3\}$, then a generic element of \eqref{uandvbis} may be written as
\begin{equation}
{(\mu_i \otimes \mathbb{I})}^{-1}  \ {(\mathbb{I} \otimes \mu_j)}^{-1} \ (\mu_i \otimes \mathbb{I}) \ (\mathbb{I} \otimes \mu_j) 
=(\mu_i^{-1} \otimes \mathbb{I}^{-1}) (\mathbb{I}^{-1} \otimes \mu_j^{-1}) (\mu_i \otimes \mathbb{I}) (\mathbb{I} \otimes \mu_j)
\end{equation}
$$=(\mu_i \otimes \mathbb{I}) (\mathbb{I} \otimes \mu_j) (\mu_i \otimes \mathbb{I}) (\mathbb{I} \otimes \mu_j)=(\mu_i \otimes \mu_j) (\mu_i \otimes \mu_j)
=(\mu_i \mu_i) \otimes (\mu_j \mu_j)
=\mathbb{I} \otimes \mathbb{I}
=\mathbb{I}
$$
and so $[U,V]$ is trivial as claimed. We may conclude that $P_2$ is central product of $U$ and $V$.
\end{proof}

We should make some comments on the ideas in the previous arguments.

\begin{remark}

It is worth noting that the specific values of $\theta$ and $\delta$ in  \eqref{specialmuoperators1} are not the only values that realise $\mu_1$, $\mu_2$ and $\mu_3$ as proportional to the Pauli matrices. We note the following examples which can be produced with a different choice of $\theta$ and $\delta$ :

\begin{equation}\label{ideaofyanga}
\mu_1(0,0)=-Z, \quad \mu_2(0,0)=iX, \quad \mu_3(0,0)=-X,
\end{equation}

In fact $\mu_1, \mu_2, \mu_3$  are proportional to  $X, Y, Z$, if $2 \theta=n \pi$ for   $n \in \mathbb{Z}$ and either  $\omega=0$  or  $\delta=0$. 

Furthermore we note explicitly that the choice of the pseudofermionic operators $\textbf{a}$ and $\textbf{b}$ in the proof of Theorem \ref{main1} has been done, in order to show analogies with the models in  \cite{baginbagbook, tripf, mostaf2, mostaf1}. 

Of course, different choices of pseudofermionic operators $\textbf{a}$ and $\textbf{b}$ can be used in the proof of Theorem \ref{main1}, or, we can keep $\textbf{a}$ and $\textbf{b}$  and  get different expressions of  \eqref{x1asmus}--\eqref{x5asmus} in terms of $\textbf{A}$, $\textbf{B}$, $\widetilde{\textbf{A}}$ and  $\widetilde{\textbf{B}}$. For instance, this last case  happens when we use \eqref{ideaofyanga} instead of \eqref{specialmuoperators1}; we  would get a different formulation of \eqref{x1asmus}--\eqref{x5asmus} always in terms of $\textbf{A}$, $\textbf{B}$, $\widetilde{\textbf{A}}$ and $\widetilde{\textbf{B}}$. 

\end{remark}

Let's note another relevant fact:

\begin{corollary} The statement of Theorem \ref{main1} (ii) is true for six operators $X_k$ which are fermionic, up to similarity in $\mathrm{GL}_4(\mathbb{C})$.
\end{corollary}
\begin{proof} Repeat the argument  of the proof of Theorem \ref{main1}, using
 Proposition \ref{similarity} along with \eqref{fa14} and \eqref{fa15}, in order to replace  $\textbf{a}$ and $\textbf{b}$   with  fermionic operators $\textbf{c}$ and $\textbf{c}^*$.
\end{proof}

Another important aspect that we isolate in form of corollary is due to the presence of the structure of central product and to its description via pseudofermionic operators:

\begin{corollary}\label{new}
There exist two subgroups $U$ and $V$ of $P_2$ such that $P_2$ is the central product of $U$ and $V$. Moreover $U$ and $V$ are generated by pseudofermionic operators.
\end{corollary}

\begin{proof} Application of Theorem \ref{main1} (iii).
\end{proof}

\section{Conclusions and analogies with other dynamical systems}

The dynamical system $\mathcal{S}$ of Theorem \ref{main1} is only one of the systems whose dynamics can be described in terms of the matrices $X_1, X_2, \ldots, X_{12}$ in \eqref{matricesx}.

Another dynamical system $\mathcal{T}$ is discussed in \cite{rame2}, involving the same matrices \eqref{matricesx} for the description of its Hamiltonian. Again, the physical system is an electronic circuit, and again the main interest is in the possibility of having a simple experimental device connected with PT symmetry in  quantum mechanics. In this case the dynamical system is described in analogy with $\mathcal{S}$ of Theorem \ref{main1}, but with some minor differences, see \cite{rame2} for more details.

What is interesting  is the fact that the operator $L_\mathcal{T}$ is now replaced by the following symmetric, but not self-adjoint, matrix:
\begin{equation}
H_\mathcal{T}=\left(
\begin{array}{cccc}
	0 & b+ir & d+ir & 0 \\
	b+ir & 0 & 0 & d-ir \\
	d+ir & 0 & 0 & d-ir  \\
	0 &d-ir  & d-ir  & 0 \\
\end{array}
\right).
\end{equation}Interestingly enough, $\{X_k \ \mid \ k=1, 2, \ldots, 12\}$ allows us to  decompose also $H_\mathcal{T}$ as follows:
\begin{equation}\label{decompositionht}
H_\mathcal{T}=\sum_{k=1}^{12}\beta_k X_k,
\end{equation}
where the only non zero $\beta_k$ are the following:
\begin{equation}
\alpha_1=d, \qquad \beta_5=\frac{b+d}{2},\qquad \beta_6=ir, \qquad \beta_{11}=\frac{b-d}{2}+ir.
\end{equation}
This means that Theorem \ref{main1} can be proved also beginning from the dynamical system $\mathcal{T}$ instead of $\mathcal{S}$. The argument of the proof of Theorem \ref{main1} applies in fact to  $X_k$ which appear also in \eqref{decompositionht}. 

More examples can be constructed from the circuits considered in  \cite{bgsa, bagpan}. Again we will find a description for the matrices \eqref{matricesx}--\eqref{matricesxbis} in terms of pseudofermionic operators, for instance, arising from the model described in \eqref{firstoperator}-\eqref{fa15}. This suggests the presence of PT symmetries both in $\mathcal{S}$ and $\mathcal{T}$, but the formalization of the PT symmetry can be difficult to write explicitly; this is made for instance in \cite{tripf}, but not in \cite{bgsa}.

Another aspect, which deserves interest, is due to the presence of the structure of central product in $P_2$; this peculiar structure has been already investigated in \cite{bagbavrus, bavuma2021}. Generally, larger Pauli groups, i.e.: $P_n$ for $n \ge 2$, may be decomposed in central products (see \cite{rocchetto}) so it may be interesting to generalize Theorem \ref{main1} to dynamical systems involving larger Pauli groups.

\end{document}